\documentclass[11pt]{article}

\usepackage{amssymb,amsthm,amsmath}
\usepackage{graphicx}
\usepackage[small]{caption}
\usepackage{subcaption}
\usepackage{epsfig}

\usepackage[utf8]{inputenc}
\usepackage{amsfonts}
\usepackage{fullpage}
\usepackage{framed}

\usepackage{enumerate}
\usepackage{enumitem}
\usepackage[hyphens]{url}
\usepackage{multirow}

\setenumerate{listparindent=\parindent}

\newtheorem{theorem}{Theorem}
\newtheorem{lemma}{Lemma}

\newtheorem{conjecture}{Conjecture}

\newcommand{\etal}{{et~al.}}
\newcommand{\ie}{{i.e.}}
\newcommand{\eg}{{e.g.}}
\newcommand{\select}{\textsc{select }}

\newcommand{\later}[1]{}
\newcommand{\old}[1]{}

\title{\textsc{Selection Algorithms with Small Groups}\footnote{A preliminary
    version of this paper appeared in the
{\em Proceedings of the 29th International Symposium on Algorithms
and Data Structures}, (WADS 2015), Victoria, Canada, August 2015,
Vol. 9214 of LNCS, Springer, pp.~189--199. }}

\author{
Ke Chen
\thanks{Department of Computer Science,
University of Wisconsin--Milwaukee, USA\@. 
Email~\texttt{kechen@uwm.edu}.}
\and
Adrian Dumitrescu\thanks{%
Department of Computer Science,
University of Wisconsin--Milwaukee, USA\@.
Email~\texttt{dumitres@uwm.edu}.}
}

\begin{document}

\maketitle

\begin{abstract}
We revisit the selection problem, namely that of computing the $i$th
order statistic of $n$ given elements, in particular the classic
deterministic algorithm by grouping and partition due to Blum, Floyd,
Pratt, Rivest, and Tarjan~(1973). 
Whereas the original algorithm uses groups of odd size at least $5$ and
runs in linear time, it has been perpetuated in the literature that using 
smaller group sizes will force the worst-case running time to become
superlinear, namely $\Omega(n \log{n})$.
We first point out that the usual arguments found in the literature justifying 
the superlinear worst-case running time fall short of proving this claim.  
We further prove that it is possible to use group size smaller than $5$
while maintaining the worst case linear running time.  
To this end we introduce three simple variants of the classic algorithm, 
the repeated step algorithm, the shifting target algorithm, and the
hyperpair algorithm,
all running in linear time.


\end{abstract}

\section{Introduction} \label{sec:intro}

Together with sorting, selection is one of the most widely used
procedures in computer algorithms. Indeed, it is easy to find numerous algorithms
(documented in at least as many research articles) that use selection as a subroutine.
Two classic examples from computational geometry are~\cite{KS86,Me85}. 

Given a sequence $A$ of $n$ numbers (usually stored in an array), 
and an integer (target) parameter $1\leq i\leq n$, 
the selection problem asks to find the $i$th smallest element in $A$. 
Sorting the numbers trivially solves the selection problem, but if one aims at a
linear time algorithm, a higher level of sophistication is needed.  
A now classic approach for selection~\cite{BFP+73,FR75,Hy76,SPP76,Y76} 
from the 1970s is to use an element in $A$ as a pivot to partition $A$
into two smaller subsequences and recurse on one of them 
with a (possibly different) selection parameter $i$.  

The time complexity of this kind of algorithms is sensitive to the
pivots used. For example, if a good pivot is used, many elements in
$A$ can be discarded; whereas if a bad pivot is used, in the worst case,
the size of the problem may be only reduced by a constant, leading to a
quadratic worst-case running time. But choosing a good pivot can be
time consuming.  

Randomly choosing the pivots yields a well-known randomized algorithm
with expected linear running time
(see \eg,~\cite[Ch.~9.2]{CLR+09},~\cite[Ch.~13.5]{KT06}, or~\cite[Ch.~3.4]{MU05}),
however its worst case running time is quadratic in $n$. 

The first deterministic linear time selection algorithm \select
(called \textsc{pick} by the authors), in fact a theoretical
breakthrough at the time, was introduced by Blum~\etal~\cite{BFP+73}.
By using the median of medians of small (constant size) disjoint groups of $A$, 
good pivots that guarantee reducing the size of the problem by a
constant fraction can be chosen with low costs. 
The authors~\cite[page~451, proof of Theorem~1]{BFP+73} required 
the group size to be at least $5$ for the \select algorithm to
run in linear time. It has been perpetuated in the literature
the idea that \select with groups of $3$ or $4$ does not run in linear time:
an exercise of the book by Cormen~\etal~\cite[page~223, exercise~{9.3-1}]{CLR+09} 
asks the readers to argue that 
``\select does not run in linear time if groups of $3$ are used''.

We first point out that the argument for the $\Omega(n \log{n})$ lower
bound in the solution to this exercise~\cite[page~23]{CLL09} 
is incomplete by failing to provide an input sequence with one third
of the elements being discarded in each recursive call in both the
current sequence and its sequence of medians; the difficulty in
completing the argument lies in the fact that these two sequences are
not disjoint thus cannot be constructed or controlled independently. 
The question whether the original \select algorithm runs in linear time
with groups of $3$ remains open at the time of this writing. 

Further, we show that this restriction on the group size is
unnecessary, namely that group sizes smaller than $5$ can be used by a
linear time deterministic algorithm for the selection problem. 
Since selecting the median in smaller groups is easier to implement
and requires fewer comparisons (\eg, $3$ comparisons for group size
$3$ versus $6$ comparisons for group size $5$), 
it is attractive to have linear time selection algorithms that use
smaller groups. Our main result concerning selection with small
group size is summarized in the following theorem.  
\begin{theorem} \label{thm:thm1}
There exist suitable variants of \select with groups of $2$, $3$, and $4$ 
running in $O(n)$ time.
\end{theorem}

\paragraph{Historical background.}
The interest in selection algorithms
has remained high over the years with many exciting
developments (\eg, lower bounds, parallel algorithms, etc) taking
place; we only cite a few here~\cite{AKSS89,BJ85,CM89,DHU+01,DZ96,DZ99,FR75,FG79,
  GKP96,HS69,Ho61,J88,Ki81,Pa96,YY82,Y76}. 
We also refer the reader to the dedicated book chapters on selection
in~\cite{AHU83,Ba88,CLR+09,DPV08,KT06,Kn98} and
the more recent articles~\cite{Al17,Ki13}, including experimental work.

\paragraph{Outline.}
In Section~\ref{sec:prelim}, the classic \select algorithm is 
introduced (rephrased) under standard simplifying assumptions. 
In Section~\ref{sec:repeat}, we introduce a variant of \textsc{select}, 
the \emph{repeated step} algorithm, which runs in linear time with
either group size $3$ and $4$.  
With groups of $3$, the algorithm executes a certain step, 
``group by $3$ and find the medians of the groups'', twice in a row.   
In Section~\ref{sec:shift}, we introduce another variant of \textsc{select}, 
the \emph{shifting target} algorithm, 
a linear time selection algorithm with group size~$4$. 
In each iteration, upper or lower medians are used based on the
current rank of the target, and the shift in the target parameter $i$ is
controlled over three consecutive iterations. 
In Section~\ref{sec:hyperpair}, we introduce yet another variant of \textsc{select}, 
the \emph{hyperpair} algorithm, 
a linear time selection algorithm with group size~$2$. 
The algorithm performs the ``group by pairs'' step four times in a row
to form hyperpairs.
In Section~\ref{sec:variants}, we briefly introduce three other variants
of \textsc{select} with group size $4$, including one due to Zwick~\cite{Z14},
all running in linear time. 

In Section~\ref{sec:experiment}, 
we compare our algorithms (with group size $3$ and $4$) 
with the original \select algorithm (with group size $5$) by deriving 
upper bounds on the exact numbers of comparisons used by each algorithm.
We also present experimental results that verify our numeric calculations.
In Section~\ref{sec:conclusion}, we summarize our results and formulate
a conjecture on the running time of the original \select algorithm from~\cite{BFP+73}
with groups of $3$ and~$4$, as suggested by our study.

\section{Preliminaries} \label{sec:prelim}

Without affecting the results, the following two standard
simplifying assumptions are convenient:
(i)~the input sequence $A$ contains $n$ distinct numbers; and 
(ii)~the floor and ceiling functions are omitted in the descriptions of
the algorithms and their analyses.
We also assume that all the grouping steps are carried out using the ``natural'' order,
\ie, given a sequence $A=\left\{a_1, a_2,\ldots, a_n\right\}$, 
``arrange $A$ into groups of size $m$'' means that group~1 contains $a_1,a_2,\ldots, a_m$,
group~2 contains $a_{m+1}, a_{m+2},\ldots, a_{2m}$ and so on.
Under these assumptions, \select with groups of $5$ (from~\cite{BFP+73}) can
be described as follows (using this group size has become increasingly
popular, see \eg,~\cite[Ch.~9.2]{CLR+09}): 
\begin{enumerate} \itemsep 1pt
\item \label{alg5:step1}
If $n\leq 5$, sort $A$ and return the $i$th smallest number.
\item
Arrange $A$ into groups of size 5.
Let $M$ be the sequence of medians of these $n/5$ groups.
Select the median of $M$ recursively, let it be $m$.
\item
Partition $A$ into two subsequences $A_1=\{x|x<m\}$ and $A_2=\{x|x>m\}$
(the order of elements is preserved).
If $i=|A_1|+1$, return $m$. 
If $i<|A_1|+1$, go to step~\ref{alg5:step1} with
$A \leftarrow A_1$ and $n\leftarrow |A_1|$. 
If $i>|A_1|+1$, go to step~\ref{alg5:step1} with $A \leftarrow A_2$,
$n\leftarrow |A_2|$ and $i\leftarrow i-|A_1|-1$.
\end{enumerate}
\begin{figure}[htbp]
\centering
\includegraphics[scale=0.6]{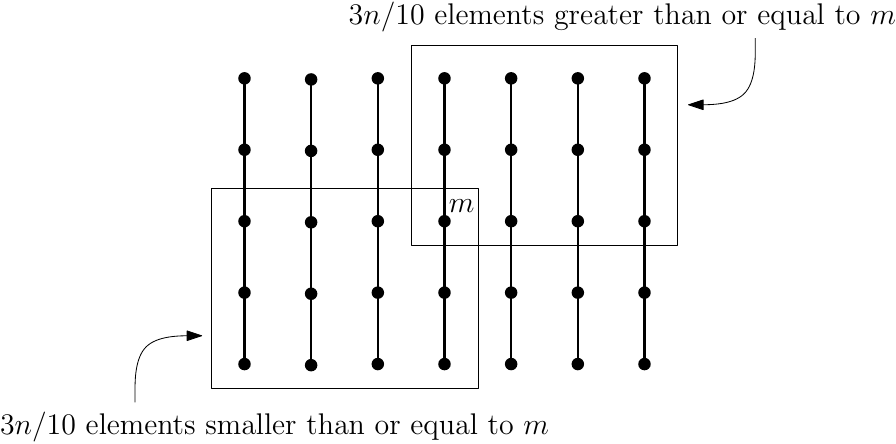}
\caption{One iteration of the \select algorithm with group size $5$.
At least $3n/10$ elements can be discarded.}
\label{fig:group5}
\end{figure}
Denote the worst case running time 
of the recursive selection algorithm on an $n$-element input by $T(n)$.
As shown in Figure~\ref{fig:group5}, at least $3*(n/5)/2=3n/10$
elements are discarded at each iteration, which yields the recurrence
\begin{equation}\label{eq:5}
  T(n)\leq T(n/5)+T(7n/10)+O(n).
\end{equation}
This recurrence is one of the following generic form:
\begin{equation} \label{eq:generic}
  T(n) \leq \sum^k_{i=1} T(a_i \, n) + O(n), \text{ where } a_i>0 \text{ for }
  i=1,\ldots,k \text{ and } \sum_{i=1}^k a_i \leq 1. 
\end{equation}
It is well-known~\cite[Ch.~4]{CLR+09}
(and can be verified by direct substitution) that the solution of~\eqref{eq:generic} is
\begin{equation}\label{eq:sol:generic}
T(n) = 
\begin{cases}
O(n) & \text{if } \sum_{i=1}^k a_i<1,\\
O(n\log n) & \text{if } \sum_{i=1}^k a_i=1.
\end{cases}
\end{equation}

As such, since the coefficients in~\eqref{eq:5} sum to $1/5+7/10=9/10<1$,
we see that the original \select algorithm with group size $5$ runs in $T(n)=\Theta(n)$
(as it is well-known).

\section{The Repeated Step Algorithm} \label{sec:repeat}

Using group size $3$ directly in the \select algorithm in \cite{BFP+73} yields
\begin{equation} \label{eq:4}
T(n)\leq T(n/3)+T(2n/3)+O(n), 
\end{equation}
which solves to $T(n)=O(n\log n)$.
Here a large portion (at least one third) of $A$ is discarded in each iteration 
but the cost of finding such a good pivot is too high, namely $T(n/3)$.
The idea of our \emph{repeated step} algorithm, inspired by the
algorithm in~\cite{BCC+00}, is to find a weaker
pivot in a faster manner by performing the operation 
``group by $3$ and find the medians'' twice in a row  
(as illustrated in Figure~\ref{fig:group3}).
It is worth noting that this method is akin to using the Tukey's ninther~\cite{Tu78}.
More precisely, $M'$ as defined in step~\ref{alg3:step2b}  below is the sequence formed
by the Tukey's ninthers of groups of $9$ elements in $A$.

\paragraph{Algorithm}
\begin{enumerate} \itemsep 1pt
\item
\label{alg3:step1}
If $n\leq 3$, sort $A$ and return the $i$th smallest number.
\item
\label{alg3:step2}
Arrange $A$ into groups of size $3$.
Let $M$ be the sequence of medians of these $n/3$ groups.
\item \label{alg3:step2b}
Arrange $M$ into groups of size $3$.
Let $M'$ be the sequence of medians of these $n/9$ groups.
\item \label{alg3:step3}
Select the median of $M'$ recursively, let it be $m$.
\item
Partition $A$ into two subsequences $A_1=\{x|x<m\}$ and $A_2=\{x|x>m\}$.
If $i=|A_1|+1$, return $m$. 
If $i<|A_1|+1$, go to step~\ref{alg3:step1} with
$A \leftarrow A_1$ and $n\leftarrow |A_1|$. 
If $i>|A_1|+1$, go to step~\ref{alg3:step1} with $A \leftarrow A_2$,
$n\leftarrow |A_2|$ and $i\leftarrow i-|A_1|-1$.
\end{enumerate}
\begin{figure}[htbp]
\centering
\includegraphics[scale=0.6]{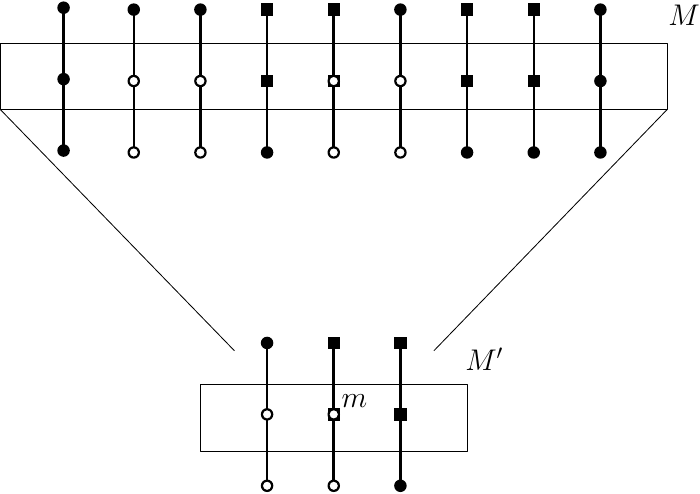}
\caption{One iteration of the \emph{repeated step} algorithm with groups of $3$. 
Empty disks represent elements that are guaranteed to be smaller than or equal to $m$.
Filled squares represent elements that are guaranteed to be greater than or equal to $m$.}
\label{fig:group3}
\end{figure}

\paragraph{Analysis.}
Since elements are discarded if and only if they are too large or too small to be
the $i$th smallest element, the correctness of the algorithm is implied. 
Regarding the time complexity of this algorithm, we have the following lemma:

\begin{lemma} \label{lem:1}
The repeated step algorithm with groups of $3$ runs in $\Theta(n)$
time on an $n$-element input. 
\end{lemma}
\begin{proof}
By finding the median of medians of medians instead of the median of medians, 
the cost of selecting the pivot $m$ reduces from $T(n/3)+O(n)$ to $T(n/9)+O(n)$.
We need to determine how well $m$ partitions $A$ in the worst case.
In step~\ref{alg3:step3}, $m$ is guaranteed to be greater than or equal to
$2*(n/9)/2=n/9$ elements in $M$. 
Each element in $M$ is a median of a group of size $3$ in $A$, so it is
greater than or equal to 2 elements in its group. 
All the groups of $A$ are disjoint, thus $m$ is greater than or equal to $2n/9$ elements
in $A$. Similarly, $m$ is smaller than or equal to $2n/9$ elements in $A$.
Thus, in the last step, at least $2n/9$ elements can be discarded.
The recursive call in step~\ref{alg3:step3} takes $T(n/9)$ time.
So the resulting recurrence is
$$ T(n)\leq T(n/9)+T(7n/9)+O(n), $$ 
and since the coefficients on the right side sum to $8/9<1$, by~\eqref{eq:sol:generic},
we have $T(n)=\Theta(n)$, as required.
\end{proof}
Note that grouping by $3$ twice and finding the median of medians of
medians is different from grouping by $9$ and finding the median of medians.
The number of comparisons required for grouping by $3$ twice is
$3n/3+3n/9=12n/9$, while for grouping by $9$ the number is $14n/9$
($14$ comparisons for selecting the median of $9$).
The number of elements guaranteed to be discarded is also different:
for grouping by $3$ twice, at least $2n/9$ elements can be discarded, while
for grouping by $9$, this number is $5n/18$.
So our method trades some of the quality of the pivots for speed (discards fewer elements
than the median of $9$ approach) by doing fewer comparisons.

\section{The Shifting Target Algorithm} \label{sec:shift}

In the \select algorithm introduced in \cite{BFP+73}, the group size is restricted
to odd numbers, where the median of a group has a privileged symmetric position.
For group size $4$, depending on the choice of upper, lower, or average median, 
there are three possible partial orders to be considered (see Figure~\ref{fig:partial}). 
\begin{figure}[htbp]
\centering
\includegraphics[scale=0.6]{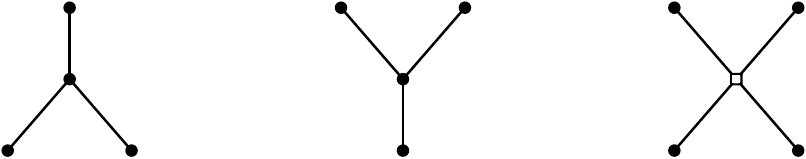}
\caption{Three partial orders of $4$ elements based on the upper (left),
  lower (middle), and average (right) medians. 
The empty square represents the average of the upper and lower
median, which is not necessarily part of the $4$-element sequence.} 
\label{fig:partial}
\end{figure}

If the upper (or lower) median is always used, only $2*(n/4)/2=n/4$ elements are guaranteed  
to be discarded in each iteration (see Figure~\ref{fig:group4_lower}),
which gives the recurrence 
\begin{equation} \label{eq:1}
T(n) \leq T(n/4)+T(3n/4)+O(n).
\end{equation}
The term $T(n/4)$ is for the recursive call to find the median of all $n/4$ medians.
This recursion solves to $T(n)=O(n\log n)$.
Even if we use the average of the two medians, the recursion remains the same
since only $2$ elements from each of the $(n/4)/2=n/8$ groups are
guaranteed to be discarded. 

Observe that if the target parameter satisfies $i \leq n/2$ (resp., $i \geq n/2$),
using the lower (resp., upper) median gives a better chance to discard
more elements and thus obtain a better recurrence; 
detailed calculations are given in the proof of Lemma~\ref{lem:2}. 
Inspired by this idea, we propose the \emph{shifting target} algorithm
as follows:
\paragraph{Algorithm}
\begin{enumerate} \itemsep 2pt
\item
\label{alg4:step1}
If $n\leq 4$, sort $A$ and return the $i$th smallest number.
\item
Arrange $A$ into groups of size $4$.
Let $M$ be the sequence of medians of these $n/4$ groups.
If $i\leq n/2$, the lower medians are used; otherwise the upper medians are used.
Select the median of $M$ recursively, let it be $m$.
\item
Partition $A$ into two subsequences $A_1=\{x|x<m\}$ and $A_2=\{x|x>m\}$.
If $i=|A_1|+1$, return $m$. 
If $i<|A_1|+1$, go to step~\ref{alg4:step1} with
$A \leftarrow A_1$ and $n\leftarrow |A_1|$. 
If $i>|A_1|+1$, go to step~\ref{alg4:step1} with $A \leftarrow A_2$,
$n\leftarrow |A_2|$ and $i\leftarrow i-|A_1|-1$.
\end{enumerate}

\paragraph{Analysis.}
Regarding the time complexity, we have the following lemma.
\begin{lemma} \label{lem:2}
The shifting target algorithm with group size $4$ runs in $\Theta(n)$
time on an $n$-element input. 
\end{lemma}
\begin{proof}
We shall prove that in at most three consecutive iterations, the size of the problem is reduced
by a large enough fraction so that the resulting recurrence is of the form in~\eqref{eq:generic}
with $\sum_{i=1}^k a_i<1$. 
\begin{figure}[htbp]
\centering
\includegraphics[scale=0.6]{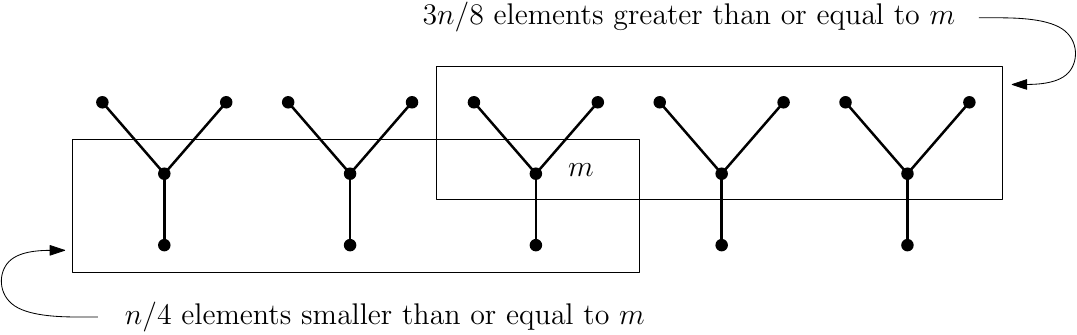}
\caption{Group size $4$ with lower medians used.}
\label{fig:group4_lower}
\end{figure}

If in some iteration, we have $i\leq n/4$, then the lower medians are used.
Recall that $m$ is guaranteed to be greater than or equal to $2*(n/4)/2=n/4$ elements of $A$. 
So either $m$ is the $i$th smallest element in $A$ or at least $3*(n/4)/2=3n/8$ largest elements
are discarded (see Figure~\ref{fig:group4_lower}), hence the worst-case running time recurrence is
\begin{equation} \label{eq:2}
T(n)\leq T(n/4)+T(5n/8)+O(n).
\end{equation} 
Observe that in this case the coefficients on the right side sum to
$7/8<1$, yielding a linear solution, as required.

Now consider the case $n/4<i\leq n/2$, again the lower medians are used.
If $|A_1| \geq i$, \ie, the rank of $m$ is higher than $i$, again at least
$3*(n/4)/2=3n/8$ largest elements are discarded and~\eqref{eq:2} applies.
Otherwise, suppose that only $t =|A_1| \geq 2*(n/4)/2=n/4$ smallest
elements are discarded. Then in the next iteration, $i'=i-t$, $n'=n-t$.

If $i'\leq n'/4$, at least $3n'/8$ elements are discarded.
The first iteration satisfies recurrence~\eqref{eq:1} and we 
can use recurrence~\eqref{eq:2} to bound the term $T(3n/4)$ from above.
We deduce that in two iterations the worst case running time satisfies the recurrence: 
\begin{align} \label{eq:3}
T(n)&\leq T(n/4)+T(3n/4)+O(n) \nonumber \\
&\leq T(n/4)+ T((3n/4)/4) + T((3n/4)*5/8) + O(n)  \nonumber \\
&= T(n/4)+T(3n/16)+T(15n/32)+O(n).
\end{align}
Observe that the coefficients on the right side sum to $29/32<1$, yielding a linear solution,
as required. Subsequently, we can therefore assume that $i' \geq n'/4$. We have
\begin{align*}
i'/n'&=(i-t)/(n-t) \leq(i-n/4)/(n-n/4)\\
&\leq(n/2-n/4)/(n-n/4) =1/3.
\end{align*}
Since $1/4<i'/n'\leq 1/3 \leq 1/2$, the lower medians will be used.
As described above, if at least $3n'/8$ largest elements are discarded, in two iterations,
the worst case running time satisfies the same recurrence~\eqref{eq:3}.

So suppose that only $t' \geq 2*(n'/4)/2=n'/4$ smallest elements are discarded.
Let $i''=i'-t'$, $n''=n'-t'$. We have
\begin{align*}
i''/n''&=(i'-t')/(n'-t') \leq(i'-n'/4)/(n'-n'/4)\\
&\leq(n'/3-n'/4)/(n'-n'/4) =1/9.
\end{align*}
Since $i''/n'' \leq 1/9 <1/4$, in the next iteration, at least $3n''/8$ elements will
be discarded. The first two iterations satisfy recurrence~\eqref{eq:1} and we 
can use recurrence~\eqref{eq:2} to bound the term $T(9n/16)$ from above.
We deduce that in three iterations the worst case running time satisfies the recurrence: 
\begin{align*}
T(n)&\leq T(n/4)+T(3n/4)+O(n)\\
&\leq T(n/4)+ T((3n/4)/4) + T((3n/4)*3/4) + O(n)\\
&= T(n/4)+T(3n/16)+T(9n/16)+O(n)\\
&\leq T(n/4)+T(3n/16)+ T((9n/16)/4) + T((9n/16)*5/8) + O(n)\\
&= T(n/4)+T(3n/16)+T(9n/64)+T(45n/128)+O(n).
\end{align*}
The sum of the coefficients on the right side is $119/128<1$,
so again by~\eqref{eq:sol:generic}, the solution is $T(n)=\Theta(n)$.

By symmetry, the analysis also holds for the case $i \geq n/2$, and the proof
of Lemma~\ref{lem:2} is complete. 
\end{proof}

\section{The Hyperpair Algorithm}\label{sec:hyperpair}

For completeness, we consider the ultimate group size $2$, \ie, each group
contains a pair of elements. The upper (resp. lower) median of a pair is the larger
(resp. smaller) element in that pair.
In the original \select algorithm, if pairs were used, only $1*(n/4)$ elements are
guaranteed to be discarded in each iteration, which gives the recurrence
\begin{equation}
T(n)\leq T(n/2)+T(3n/4)+O(n).
\end{equation}

The term $T(n/2)$ is for the recursive call to find the median of the $n/2$ upper 
(or lower) medians. However, the above recursion does not yield a solution linear in $n$.
Now, one can make the following adjustment: instead of taking the median of half the input 
recursively, let the algorithm recursively compute
the $j$th smallest element among the $n/2$ upper medians, where $j=n/6$. 
Then $2j=n/2-j=n/3$ elements can be discarded in each iteration, thus the size of the
largest remaining recursive call is $n-n/3=2n/3$. However, even with this adjustment,
the resulting recurrence~\eqref{eq:hp2} does not yield a solution linear in $n$.
\begin{equation} \label{eq:hp2}
T(n)\leq T(n/2)+T(2n/3)+O(n).
\end{equation}

The key for obtaining a linear running time in this setting
seems to be to use groups of $2$ in a repeated manner.
The following algorithm has the same flavor as the repeated step algorithm in 
section~\ref{sec:repeat} but uses group size $2$.
Its name, the \emph{hyperpair} algorithm, will be justified in the analysis.

\paragraph{Algorithm}
\begin{enumerate} \itemsep 1pt
\item
\label{alg2:step1}
If $n\leq 2$, sort $A$ and return the $i$th smallest number.
\item
\label{alg2:step2}
Arrange $A$ into groups of size $2$.
Let $M_1$ be the sequence of upper medians of these $n/2$ pairs.
\item
\label{alg2:step3}
Arrange $M_1$ into pairs.
Let $M_2$ be the sequence of lower medians of these $n/4$ pairs.
\item
\label{alg2:step4}
Arrange $M_2$ into pairs.
Let $M_3$ be the sequence of upper medians of these $n/8$ pairs.
\item
\label{alg2:step5}
Arrange $M_3$ into pairs.
Let $M_4$ be the sequence of lower medians of these $n/16$ pairs.
\item
Select the median of $M_4$ recursively, let it be $m$.
\item
Partition $A$ into two subsequences $A_1=\{x|x<m\}$ and $A_2=\{x|x>m\}$.
If $i=|A_1|+1$, return $m$. 
If $i<|A_1|+1$, go to step~\ref{alg2:step1} with
$A \leftarrow A_1$ and $n\leftarrow |A_1|$. 
If $i>|A_1|+1$, go to step~\ref{alg2:step1} with $A \leftarrow A_2$,
$n\leftarrow |A_2|$ and $i\leftarrow i-|A_1|-1$.
\end{enumerate}

\paragraph{Analysis.}
In order to calculate the time complexity of this algorithm, we need to estimate how
well $m$ partitions the sequence $A$.
Observe that steps~\ref{alg2:step2}--\ref{alg2:step5} can be viewed as constructing
\emph{hyperpairs}, as in the non-recursive selection algorithm of
Sch\"{o}nhage~\etal~\cite{SPP76}. In their definition,
a single element is a hyperpair with itself as the \emph{center};
given two disjoint copies of a hyperpair, we can combine them to form a larger
hyperpair by comparing their centers and taking the upper or lower of these
as the new center. The hyperpairs $P$ constructed in our algorithm are illustrated
in Figure~\ref{fig:group2}.
\begin{figure}[htbp]
\centering
\includegraphics[scale=0.5]{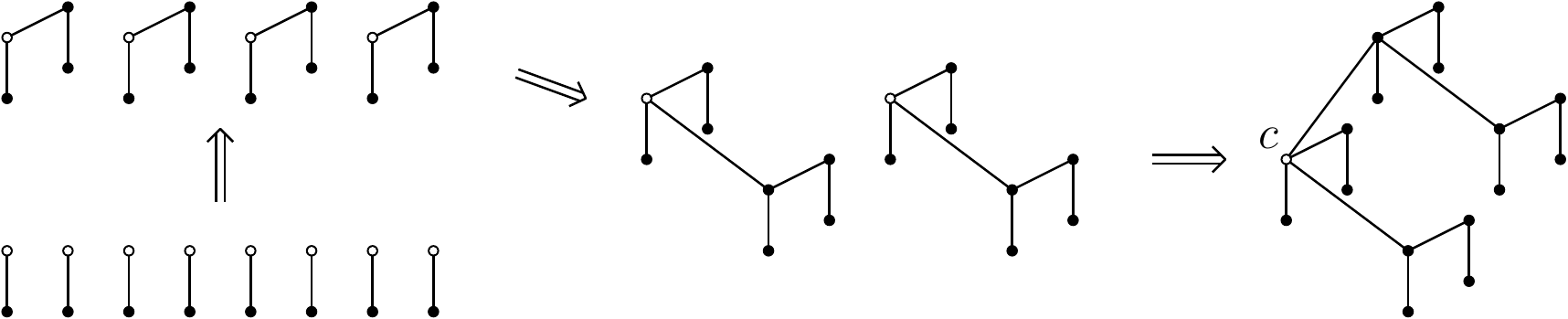}
\caption{Construction of a hyperpair $P$ with $16$ elements; the center of each
hyperpair is marked by an empty circle.}
\label{fig:group2}
\end{figure}
Observe that in $P$, three elements are guaranteed to be greater than its center $c$ and 
three are guaranteed to be smaller than $c$. We are now ready to establish the time complexity
of this algorithm:

\begin{lemma} \label{lem:hyperpair}
The hyperpair algorithm runs in $\Theta(n)$
time on an $n$-element input. 
\end{lemma}
\begin{proof}
Steps~\ref{alg2:step2}--\ref{alg2:step5} take $n/2+n/4+n/8+n/16=15n/16$ comparisons to form
the hyperpairs $P$.
The pivot $m$ is the median of the centers of these $n/16$ hyperpairs.
So the cost of selecting the pivot is $T(n/16)+15n/16$.
By the above observation about the center $c$ of $P$, $m$ is guaranteed to be greater than 
or equal to $4*(n/16)/2=n/8$ elements in $A$. 
Similarly, $m$ is guaranteed to be smaller than or equal to $n/8$ elements in $A$.
Thus, in the last step, at least $n/8$ elements can be discarded.
The resulting recurrence is
$$ T(n)\leq T(n/16)+T(7n/8)+O(n), $$ 
and since the coefficients on the right side sum to $15/16<1$, by~\eqref{eq:sol:generic},
we have $T(n)=\Theta(n)$, as required.
\end{proof}

Note that larger hyperpairs can also be used to obtain linear-time algorithms.
If the ``group into pairs''
step is repeated $2k$ times, $k\geq 2$, where upper and lower medians are used alternatively,
then $n/2^{2k}$ hyperpairs of size $2^{2k}$ are built. Each center is guaranteed to be greater 
than or equal to $2^k$ elements in its hyperpair and is also guaranteed to be smaller than or 
equal to $2^k$ elements in its hyperpair. So using the median of these centers as pivot, at least 
$2^k*\left(n/2^{2k}\right)/2=n/2^{k+1}$ elements can be discarded. The resulting recurrence is
\[
T(n)\leq T\left(n/2^{2k}\right)+T\left(\left(1-1/2^{k+1}\right)n\right)+O(n),
\]
where the $O(n)$ term involves $\sum_{j=1}^{2k} n/2^j =n - n/2^{2k}$ comparisons to build
the hyperpairs and at most $n$ comparisons to partition the sequence. 
Since the coefficients on the right side sum to $1-\left(2^{k-1}-1\right)/2^{2k}<1$, 
by~\eqref{eq:sol:generic}, we have $T(n)=\Theta(n)$.

\section{Other Variants}\label{sec:variants}

A similar idea of repeating the group step (from Section~\ref{sec:repeat})
also applies to the case of groups of $4$ and yields 
$$ T(n)\leq T(n/16)+T(7n/8)+O(n),$$ 
and thereby another linear time selection algorithm with group size $4$. 

\paragraph{A hybrid algorithm.}
Yet another variant of \select with group size $4$ (we refer to it as
the hybrid algorithm), can be obtained by using the ideas of both
algorithms together, \ie, repeat the grouping by $4$ step twice
in a row while $M$ contains the lower medians and $M'$ contains the
upper medians (or vice versa). 
Recursively selecting the median $m$ of $M'$ takes time $T(n/16)$.
Notice that $m$ is greater than or equal to $3*(n/16)/2=3n/32$ elements in $M$ 
of which each is greater than or equal to 2 elements in its group in $A$.
So $m$ is greater than or equal to $3n/16$ elements of $A$.
Also, $m$ is smaller than or equal to $2*(n/16)/2=n/16$ elements in $M$
of which each is smaller than or equal to $3$ elements in its group of $A$.
So $m$ is smaller than or equal to $3n/16$ elements of $A$, thus
the resulting recurrence is 
$$T(n)\leq T(n/16)+T(13n/16)+O(n), $$
again with a linear solution, as desired. 

\paragraph{Zwick's variant.}
The fact that the \select algorithm can be modified so that it works
with groups of $4$ in linear time was observed prior to this writing. 
The following variant, from 2010, is due to Zwick~\cite{Z14}. 
Split the elements of $A$ into quartets.
Find the $2$nd smallest element of each quartet (\ie, the lower median), and let 
$M$ be this subset of $n/4$ elements.
Recursively find the $(3/5)(n/4)$th smallest element $m$ of $M$. 
Now $(3/5)(n/4)$ groups of $A$ have $2$ elements smaller than or equal to $m$, 
so $m$ is greater than or equal to $2(3/5)(n/4)=3n/10$ elements in
$A$. Similarly, $(2/5)(n/4)$ groups of $A$ have $3$ elements greater than or equal to $m$, 
so $m$ is smaller than or equal to $3(2/5)(n/4)=3n/10$ elements in $A$. 
Thus, the remaining recursive call involves at most $7n/10$ elements,
and the resulting recurrence is 
$$T(n) \leq T(n/4) + T(7n/10) + O(n). $$
Since $1/4+7/10<1$, the solution is linear.

\section{Comparison of the Algorithms and Experimental Results}
\label{sec:experiment}

To compare our algorithms with the original \select algorithm,
we first derive upper bounds on the exact numbers of comparisons for each
variant in the same manner as in Section~2 of~\cite{BFP+73}.
It should be noted that all recurrent formulas and all proofs do not provide 
(nor aim to provide) tight bounds or expected number of comparisons. 
Tighter analytical bounds might exist than those shown.
Let now $T(n)$ denote the total number of comparisons performed. 
For the original \select algorithm with group size $5$, we have
$$T(n)\leq T(n/5)+T(7n/10)+6n/5+n,$$
in which the term $6n/5$ is for computing the $n/5$ medians (each takes at most $6$ comparisons)
and the term $n$ is for partitioning the sequence around the selected pivot.
Solving the recurrence yields $T(n) \leq 22n$.
Similarly, for the repeated step algorithm, we have
$$T(n)\leq T(n/9)+T(7n/9)+3n/3+3n/9+n,$$
and consequently, $T(n) \leq 21n$.
For the hybrid algorithm, we have
$$T(n) \leq T(n/16) + T(13n/16)+ 4n/4 +4n/16 + n,$$
and consequently, $T(n)\leq 18n$.
For Zwick's algorithm, we have
$$T(n)\leq T(n/4)+T(7n/10)+4n/4+n,$$
and consequently, $T(n)\leq 40n$.
For the hyperpair algorithm, we have
$$T(n)\leq T(n/16)+T(7n/8)+15n/16+n,$$
and consequently, $T(n)\leq 31n$.
For the shifting target algorithm, the analysis is more involved; 
it yields $T(n) \leq 66n$.

\begin{center}
\begin{table} [hbtp]
\begin{tabular}{|c|c|c|c|c|c|c|c|}
\hline
\multirow{2}{*}{Algorithm} & \multirow{2}{*}{Group} & \multirow{2}{*}{Upper Bound}
& \multirow{2}{*}{Average Time} & \multicolumn{2}{|c|}{Comparisons} & \multicolumn{2}{|c|}{Swaps}\\
\cline{5-8}
& & & & Average & Max & Average & Max\\
\hline\hline
Hybrid & 4 & $18n$ & 364.3ms & 4.1 & 4.2 & 1.2 & 1.2\\
\hline
Repeated step& 3 & $21n$ & 446.9ms & 4.3 & 4.4 & 1.8 & 1.8\\
\hline
Original & 5 & $22n$ & 468.9ms & 5.7 & 5.8 & 1.5 & 1.5\\
\hline
Hyperpair(4) & 2 & $31n$ & 480.6ms & 2.9 & 2.9 & 3.0 & 3.0\\
\hline
Zwick's & 4 & $40n$ & 541.1ms & 6.3 & 6.3 & 2.0 & 2.0\\
\hline
Shifting target & 4 & $66n$ & 558.0ms & 6.6 & 6.7 & 2.0 & 2.1\\
\hline
Original & 4 & $O(n\log n)$ & 560.2ms & 6.7 & 6.7 & 2.0 & 2.0\\
\hline
Original & 3 & $O(n\log n)$ & 813.4ms & 8.2 & 8.5 & 3.4 & 3.5\\
\hline
\hline
Hyperpair(6) & 2 & $127n/3$ & 452.4ms & 2.8 & 2.8 & 2.8 & 2.8\\
\hline
Hyperpair(8) & 2 & $73n$ & 456.0ms & 2.8 & 2.8 & 2.8 & 2.9\\
\hline
Hyperpair(10) & 2 & $2047n/15$ & 458.8ms & 2.9 & 2.9 & 2.9 & 2.9\\
\hline
\end{tabular}
\caption{Experimental results. 
The last four columns are values per element.
The numbers in parentheses for the hyperpair algorithms indicate
the numbers of times the ``group into pairs'' step is repeated. 
The ``Upper Bound'' column shows the leading term in the solution of the corresponding
recurrence for the worst-case number of comparisons.}
\label{tab:exp}
\end{table}
\end{center}

\vspace{-7mm}
We note that sharper upper bounds are possible by taking extra care in avoiding comparisons
with known outcomes against the pivot; however, for simplicity of implementation 
we opted to forego this saving.
In order to avoid the overhead of repeated array copying, 
all the algorithms were implemented in-place, in the sense that,
with the exception of the recursion, only $O(1)$ extra space 
is used in addition to the input array. This requires minor
modifications of the algorithms; however, their running time analyses
remain unchanged. 
We carried out 1000 experiments\footnote{The experiments
were performed on a desktop with 64bits operating system, 7.8GB memory and 
Intel\textregistered\;Core\texttrademark\;i7-2600 3.4GHz processor.
The C code used can be downloaded at
\url{https://drive.google.com/file/d/0B7USj6ZPkysnMjNwV014RDJGMWc/view?usp=sharing}.}
on selecting medians in arrays of 10 million randomly permuted distinct integers.
The results are summarized in Table~\ref{tab:exp}.

We observed that the experimental results agree with the worst-case estimates in the
number of comparisons, in the sense that they show roughly the same speed ranking.
One reason why the experimental speed ranking does not fully match the 
analytical bounds derived is the existence of other operations performed during
the selection process that are unaccounted for by the recurrences, such as data copying
(shown in the last two columns of the table as swaps).
It is worth noting that optimizations introduced in Section~3 of~\cite{BFP+73},
or others discussed in~\cite{Al17},
may be used to reduce the multiplicative constant factors.

\section{Conclusion} \label{sec:conclusion}

The question whether the original selection algorithm introduced
in~\cite{BFP+73} (outlined in Section~\ref{sec:prelim}) 
runs in linear time with group size $3$ and $4$ remains unsettled.
Although the recurrences 
\begin{align*}
T(n)&\leq T(n/3)+T(2n/3)+O(n)\text{, and}\\
T(n)&\leq T(n/4)+T(3n/4)+O(n)
\end{align*}
(see~\eqref{eq:4} and~\eqref{eq:1}) for its worst-case running time
with these group sizes both solve to $T(n)=O(n\log n)$,  
we believe that they only give non-tight upper bounds on the worst case scenarios.
In any case and against popular belief we think that $\Theta(n \log{n})$ is \emph{not}
the answer in regard to the time complexity of selection with these group sizes:
\begin{conjecture}
The \select algorithm introduced by Blum~\etal~\cite{BFP+73} 
runs in $o(n \log{n})$ %linear 
time with groups of $3$ or $4$. 
\end{conjecture}

\end{document}